\documentclass[twoside]{article}
\pdfoutput=1
\usepackage{a4wide}
\usepackage{tabularx}
\usepackage{amsfonts}
\usepackage{amssymb}
\usepackage{amsmath}
\usepackage[amsmath,thmmarks]{ntheorem}
\usepackage{mathtools}
\usepackage[mathcal]{euscript}
\usepackage{proof}
\usepackage{stmaryrd}

\date{\today}
\author{Assal\'e Adj\'e and Jean Goubault-Larrecq}
\title{Concrete Semantics of Programs with Non-Deterministic \\ and Random Inputs}

\def\rr{\mathbb{R}}
\def\qq{\mathbb{Q}}
\def\br{\rr_e}
\def\brp{\overline{\rr}_{+}}
\def\nn{\mathbb{N}}

\def\ff{\mathbb{F}}
\def\bf{\ff_e}

\def\proj{\mathbf{proj}_{\bf}}
\def\inj{\mathbf{inj}}
\def\err{\mathbf{err}}
\def\flmi{\mathbf{F}_{\text{min}}}
\def\flmx{\mathbf{F}_{\text{max}}}

\def\semr#1{\llbracket #1\rrbracket_r}
\def\semf#1{\llbracket #1\rrbracket_f}
\def\seme#1{\llbracket #1\rrbracket_{\star}}

\def\esigma{\Sigma_{\star}}
\def\rhos{\rho_{\star}}
\def\Va{\mathcal{V}}
\def\db#1{\dot{#1}}
\def\lsupexp#1#2{\prescript{#1}{}{#2}}

\def\lfp#1{\operatorname{lfp}\left(#1\right)}
\def\part{\mathbb{P}}
\def\ifp{\operatorname{if}}
\def\whilep{\operatorname{while}}
\def\elsep{\operatorname{else}}
\def\thenp{\operatorname{then}}
\def\skipp{\operatorname{skip}}
\def\weak#1{\operatorname{wp}\seme{#1}}

\def\fun{\operatorname{fun}}

\def\spesem#1{\llbracket #1\rrbracket_c}

\def\choqint#1{\mathcal{C}\!\!\!\!\!\int_{#1}}
\def\re{\rr_e}
\def\fe{\ff_e}
\def\bor#1{\mathcal{B}\left(#1\right)}

\def\mesp#1{\mathcal{M}_{+}\!\!\left(#1\right)}

\def\supi{\operatorname*{\sup\nolimits^\uparrow}}
\def\bigcupi{\operatorname*{\bigcup\nolimits^\uparrow}}

\newcommand\NAN{\mathtt{NaN}}
\newcommand\pinf{\mathtt{inf}}
\newcommand\minf{-\pinf}
\newcommand\intT{\mathtt{int}}

\newcommand\ANS{\mathtt{Ans}}

\newcommand\pow{\mathbb{P}}
\renewcommand\qed{\hfill$\Box$}

{\theoremstyle{plain}
\theorembodyfont{\upshape}

}

{\theoremstyle{break}
\theorembodyfont{\upshape}
\newtheorem{convention}{Convention}
}

{\theoremstyle{break}
\theorembodyfont{\itshape}
\newtheorem{proposition}{Proposition}
}

{\theoremstyle{break}
\theorembodyfont{\itshape}

}

{\theoremstyle{break}
\theorembodyfont{\itshape}

}

{\theoremstyle{break}
\theorembodyfont{\itshape}
\newtheorem{lemma}{Lemma}
}

{\theoremstyle{break}
\theorembodyfont{\itshape}

}

{\theoremstyle{break}
\theorembodyfont{\itshape}

}

{\theoremstyle{nonumberplain}
\theorembodyfont{\upshape}
\newtheorem{proof}{Proof}
}

{\theoremstyle{break}
\theorembodyfont{\itshape}
\newtheorem{defi}{Definition}
}

{\theoremstyle{break}
\theorembodyfont{\itshape}
\newtheorem{assum}{Assumption}
}
\everymath{\displaystyle}

\begin{document}
\maketitle

\begin{abstract}
  This document gives semantics to programs written in a C-like
  programming language, featuring interactions with an external
  environment with noisy and imprecise data.


\end{abstract}


\bibliographystyle{alpha}

\section{Introduction}
\label{sec:intro}

The purpose of this report is to define a concrete semantics for a toy
imperative language, meant to incorporate the essential features of
languages such as C, as used in numerical control programs such as
those used in the ANR CPP project.

Some of the distinctive aspects of these programs are: the prominent
use of floating-point operations; and the fact that these programs
read inputs from sensors.  Both these features imply that the values
of numerical program variables are \emph{uncertain}.  Floating-point
operations are vulnerable to round-off errors, which can be modeled as
quantization noise.  Uncertainty is probably more manifest with
sensors, which return values up to some measurement error.  This
measurement error can be described by giving guaranteed bounds (this
is \emph{non-determinism}: any value in the interval can be the actual
value), or by giving a probability distribution (this is
\emph{randomness}: some values are more likely than others), or a
combination of both.  To deal with the latter, more complex
combinations, we rest on variants of two semantic constructions that
were studied by the first author, \emph{previsions} \cite{Gou-csl07}
and \emph{capacities} \cite{JGL:icalp07:cred}.

The main goal of a concrete semantics is to serve as a reference.  In
our case, we wish to be able to prove the validity of associated
abstract semantics and static analysis algorithms, as presented in
other CPP deliverables.  The kind of abstract semantics we are
thinking of was produced, as part of CPP, in \cite{BGGLP-scan10}.
(While it might seem strange that the publication of the abstract
semantics predates the design of the concrete semantics, one might say
that both were developed at roughly the same time, with an eye on each
other.)  So one of our constraints was to ensure that our concrete
semantics should make it easy to justify the abstract semantics we
intend.

Before we start, we should also mention an important point.  Numerical
programs manipulate floating-point values, which are values from a
finite set meant to denote some approximate real values.  It is
customary to think of floating-point values as reals, up
to some error.  This is why we shall define a first semantics, called
the \emph{real semantics}, where variables hold actual reals, and no
round-off is performed at all.  This has well-defined mathematical
contents, but is not what genuine C programs compute.  So we define a
second semantics, the \emph{floating-point semantics}, which is meant
to faithfully denote what C programs compute, but works on
floating-point data, mathematically an extremely awkward concept:
e.g., floating-point addition is not associative, has one absorbing
element ($\NAN$), has no inverse in general (the opposite of infinity
$\pinf$, $\minf$, is not an inverse since the sum of $\pinf$ and $\minf$ is
$\NAN$, not $0$).  But the two semantics are related, through
\emph{quantization}, which is roughly the process of rounding a real
number to the nearest floating-point value.

While the real semantics is much simpler to define than the
floating-point semantics without random choice or non-determinism
(e.g., the semantics of $+$ is merely addition), the situation changes
completely in the presence of random or non-deterministic choice.  Let
us explain this briefly.  The prevision semantics of the style
presented in \cite{Gou-csl07} is based on \emph{continuous} maps and
continuous previsions.  This is perfectly coherent for ordinary,
non-numerical programs (or for numerical programs in the
floating-point semantics, where the type of floating-point numbers is
merely yet another finite data type).  However, this is completely at
odds with the real semantics.  To give a glimpse of the difficulty,
one can define the Heaviside function $\chi_{[0, +\infty)}$ as a
numerical C program with real semantics, say by $\ifp\ x<0\ \thenp\
0.0\ \elsep\ 1.0$, and this is definitely not continuous.  The deep
problem is that, up to some inaccuracies, continuous semantics cannot
describe more than computable operations, but the real semantics must
be \emph{non-computable}: even if we restricted ourselves to
computable reals, testing whether a computable real is equal to $0$ is
undecidable.

There are at least two ways to resolve this conundrum.  The first one
is to cling to the continuous semantics of random choice and
non-determinism of \cite{Gou-csl07} or \cite{JGL:icalp07:cred}, and
work not on reals (or tuples of reals, in $\rr^n$, representing the
list of values of all $n$ program variables), but rather on a
\emph{computational model} of $\rr^n$.  The notion of computational
model of a topological space originates from Lawson \cite{Lawson:Max}.
For example, the dcpo of non-empty closed intervals of reals is a
computational model for $\rr$, and the Heaviside map would naturally
be modeled as the function mapping every negative real to $0$, every
positive real to $1$, and $0$ to the interval $[0, 1]$.  This is
elegant, mathematically well-founded, and would allow us to reuse the
continuous constructions of \cite{Gou-csl07} or
\cite{JGL:icalp07:cred}.  But it falls short of giving an account of
real number computation as operating on \emph{reals}.

We shall explore the second way here: we shall give a real semantics
in terms of \emph{measurable}, not continuous, maps.  This will give
us the required degrees of freedom to define our semantics---e.g., the
Heaviside map \emph{is} measurable---while allowing us to define the
semantics of random, non-deterministic and mixed choice: anticipating
slightly on future sections, this involves generalized forms of
integration, which will be well-defined precisely on measurable maps.
We develop the required theory in sections to come, by analogy with
both the classical Lebesgue theory of integration and the above cited
work on continuous previsions and capacities.

In the presence of random choice only (no non-determinism), our
semantics will be isomorphic to Kozen's semantics of probabilistic
programs \cite{Kozen:probabilistic}, and his clauses for computing
expectations backwards will match our prevision-based semantics.
The semantics we shall describe in the presence of other forms of
choice (non-deterministic, mixed) are new.

\paragraph{Outline.} 
In Section~\ref{sec:syntax}, we introduce the syntax of the programs 
analyzed. In Section~\ref{sec:conversion}, we define the maps which are used 
to pass from floating points to real numbers and vice versa. In 
Section~\ref{sec:concrete}, we define the concrete semantics of expressions 
and tests and prove the measurability of the semantics. In Section~\ref{sec:weakestconcret},
we define our concrete semantics based as a continuation-passing semantics. We also prove 
in Section~\ref{sec:weakestconcret}, the link between our semantics and the theory of previsions.
Finally, in Section~\ref{sec:input}, we treat separately the semantics of the instructions \emph{input}.

\section{Syntax}
\label{sec:syntax}

Let $\Va$ be a countable set of so-called (program) variables.  For
each operation $op$ on real numbers, we reserve the symbol $\db{op}$
for a syntactic operation meant to implement $op$ (in the real
semantics) or some approximation of $op$ (in the floating-point
semantics).  The syntax of a simple imperative language working on
real/floating-point values is given in Figure~\ref{grammaire}.  This
syntax does not include any non-deterministic or probabilistic choice
construct: uncertainty will be in the initial values of the variables,
and will not be created by the program while running.

\begin{figure}[ht!]
\begin{center}
  \[
  \boxed{
    \begin{array}{ccll}
      expr&::=& a&a\in\qq\\      
      &|& x&x\in\Va\\
      &|& \db{-}expr&\\
      &|& expr\db{+}expr&\\
      &|& expr\db{-}expr&\\
      &|& expr\db{\times}expr&\\
      &|& expr\db{/}expr&\\
      &&&\\
      test&::=&expr\db{<=}expr&\\
      &|& expr\db{<}expr&\\
      &|& expr\db{==}expr&\\
      &|& expr\db{!=}expr&\\
      &|& !test&\\
      &&&\\
      inst&::=& \lsupexp{\ell}{skip}&\\
      &|& \lsupexp{\ell}{x =expr}&x\in\Va\\
      &|& \lsupexp{\ell}{\ifp\ test\ \thenp\ \{inst\}\ \elsep\ \{inst\}}&\\
      &|& \lsupexp{\ell}{\whilep\ test\ \{inst\}}&\\ 
      &|& inst\ \mathbf{;}\ inst&\\
    \end{array}
  }
  \]
\caption{Syntax of Programs}
\label{grammaire}
\end{center}
\end{figure}

\section{Conversion between Floating Point and Real Numbers}
\label{sec:conversion}

We shall consider two different semantics in
Section~\ref{sec:concrete}.  The first one implements arithmetic with
floating-point numbers, while the second one relies on actual real
numbers.  Here, we describe the two types and how we convert between
them.

However, one should first be aware of the pitfalls that are hidden in
such a task \cite{Monniaux:pitfalls}.  First and foremost,
floating-point numbers are meant to give approximations to real
numbers, but floating-point computations may give values that are
arbitrarily far from the corresponding real number computation.
Monniaux (op. cit., Section~5) gives the example of the following
program:
\begin{verbatim}
double modulo(double x, double mini, double maxi) {
  double delta = maxi-mini;
  double decl = x-mini;
  double q = decl/delta;
  return x - floor(q)*delta
}

int main() {
  double m = 180.;
  double r = modulo(nextafter(m,0.), -m, m);
}
\end{verbatim}
In a semantics working on real numbers, \verb/modulo/ would return the
unique number $z$ in the interval $[\verb/mini/, \verb/maxi/)$ such
that $\verb/x/-z$ is a multiple of the interval length
$\verb/maxi/-\verb/mini/$.  So, certainly, whatever \verb/nextafter/
actually computes, \verb/r/ should be in the interval $[-180, 180)$.

However, running this using IEEE 754 floating-point arithmetic may
(and usually will) return $-180.0000000000000284$ for $\verb/r/$.
(Here we need to say that \verb/nextafter(m,0.)/ returns the
floating-point that is maximal among those that are strictly smaller
than \verb/m/.  This has no equivalent in the world of real numbers,
and accordingly our language does not include this function.)  This is
only logical:
\begin{itemize}
\item When we enter \verb/modulo/, \verb/x/ is equal to $180 - 2^{-45}$;
\item Then $\verb/maxi/-\verb/mini/$ is computed ($=360$), and
  $\verb/x/-\verb/mini/$ is computed ($=360 - 2^{-45}$); these values
  are then rounded to the nearest floating-point number, and this is
  $360$ in \emph{both} cases;
\item so \verb/delta/, \verb/decl/ are both equal to $360$, \verb/q/
  equals $1$;
\item so \verb/modulo/ returns (the result of rounding applied to)
  $(180 - 2^{-45}) - (1 \times 360) = -180 - 2^{-45} \simeq
  -180.0000000000000284$.
\end{itemize}
Of course, the right result, if computed using real numbers instead of
floating-point numbers, should be $180 - 2^{-45} \simeq
179.9999999999999716$.

This example can be taken as an illustration of the fact that, even
though one can think of each \emph{single} operation (addition,
product, etc.) as being implemented in floating-point computation as
though one first computed the exact, real number result first, and
then rounded it, hence obtaining a best possible approximant, this is
no longer true for whole programs.

Monniaux goes further, and stresses the fact that various choices in
compiler options (e.g., \verb/x87/ vs. IEEE 754 arithmetic), IEEE 754
rounding modes, abusive optimization strategies (e.g., where the
compiler uses the fact that addition is associative, which is wrong in
floating-point arithmetic, see op.\ cit., Section~4.3.2),
processor-dependent optimization strategies (e.g., see op.\ cit.,
Section~3.2, about the use of the multiply-and-add assembler
instruction on PowerPC microprocessors), pragmas (op.\ cit.,
Section~4.3.1), all may result in surprising changes in computed
values.

This causes difficulties in defining sound semantics for
floating-point programs, discussed in op.\ cit., Section~7.3.

But our purpose is not to verify arbitrary numerical programs, and one
can make some simplifying assumptions:
\begin{enumerate}
\item We assume that floating-point arithmetic is performed using the
  IEEE 754 standard on floating-point values of a standard, fixed
  size, typically the 64-bit IEEE 754 (``\verb/double/'') type.  By
  this, we not only mean that the basic primitives are implemented as
  the standard prescribes, but that all floating-point values are
  stored in this format, even when stored in registers.  This is meant
  to avoid the sundry, dreaded problems mentioned by Monniaux with the
  use of \verb/x87/ arithmetic (where registers hold 80-bit
  intermediate values).
\item We assume that the rounding mode is fixed, once and for all for
  all programs.  In particular, calls to functions that change the
  rounding mode on the fly are prohibited.
\item We assume that all optimizations related to floating-point
  computations are turned off.  This is meant to avoid abusive
  (unsound) optimizations (e.g., assuming associativity), and also to
  avoid processor-dependent optimizations (e.g., compiling $a \times x
  + b$ using a single multiply-and-add instruction: this skips the
  intermediate rounding that should have occurred when computing $a
  \times x$, and therefore changes the floating-point semantics).
\item\label{assum:nonextafter} We assume that the only floating-point
  operations allowed are arithmetic operations (i.e., $\db{+}$,
  $\db{-}$, $\db{\times}$, $\db{/}$, but not \verb/nextafter/ for
  example, or the \verb/%f/, \verb/%g/ and related directives of
  \verb/printf/, \verb/scanf/ and relatives; nor casts to and from the
  $\intT$ type---which we shall actually omit).  Library functions
  such as \verb/sin/, \verb/cos/, \verb/exp/, \verb/log/ would be
  allowable in principle, and their semantics would follow the same
  ideas as presented below---provided we make sure that their
  implementations produce results that are correct in the ulp as
  well (i.e., that they are computed as though the exact result was
  computed, then rounded; the \emph{ulp}, a.k.a., the unit in the last
  place, is the least significant bit of the mantissa). 
\end{enumerate}
These assumptions allow us to simplify our semantics considerably.

Let us go on with the actual data types of floating-point, resp.\ real
numbers.  The IEEE 754 standard specifies that, in addition to values
representing real numbers, floating-point values include values
denoting $+\infty$ (which we write $\pinf$), $-\infty$ ($\minf$), and
silent errors ($\NAN$, for ``not a number'').  One can obtain the
first two through arithmetic overflow, e.g., by computing $1.0/0.0$ or
$-1.0/0.0$, and $\NAN$s, e.g., by computing $\pinf - \pinf$.  Under
Assumption~(\ref{assum:nonextafter}) above, there will be no way of
distinguishing any such values through the execution of expressions.
We abstract them all into a unique symbol $\err$ (\textit{error}).

An added benefit of this abstraction is that it dispenses us from
considering the difference between the two zeroes, $+0.0$ and $-0.0$,
of IEEE 754 arithmetic.  These are meant to satisfy $1.0 / \pinf =
+0.0$, $1.0 / \minf = -0.0$, but are otherwise equal, in the sense that
the equality predicate applied to $+0.0$ and $-0.0$ must return true.
Collapsing $\pinf$, $\minf$, and $\NAN$ into just one value $\err$
therefore also allows us to confuse the two zeroes, without harm.
This is important if we stick to our option that single floating-point
operations should computed the exact result then round: rounding the
real number $0$ to the nearest would be a nonsense with two
floating-point numbers representing $0$.

The error $\err$ is absorbing for all standard arithmetic
operations. This means that in our semantic definitions we assume that
the error is propagated during the execution of a program which
contains these special numbers.  Now, we extend by the error symbol
$\err$ the classical sets of floating points and real numbers, which
we denote respectively by $\ff$ and $\rr$.  This yields two new sets:
$\bf=\ff\cup \{\err\}$ and $\br=\rr\cup \{\err\}$.

\begin{convention}
\label{convention1}
Let $r$ be in $\br$. Let $\diamond$ be in $\{+,-,\times,/\}$. Then:
\[
\err\diamond r=r\diamond \err=r/0=-\err=\err
\]
\end{convention}

We consider floating point as special real numbers.  Formally, there
is a canonical injection $\inj$ that lets us to convert a
floating-point value (in $\bf$) into a real number (in $\br$):
\[
\begin{array}{cccc}
\inj:&\bf &\to &\br\\
& f&\mapsto &\inj(f)=\left\{ \begin{array}{cr}\err & \text{ if } f=\err\\
f & \text{ otherwise }
\end{array}\right. 
\end{array}
\]

Conversely, there is a projection map $\proj : \br \to \bf$ that
converts a real number to its rounded, floating-point representation,
as follows.  We let $\flmi$ be the smallest floating point number and
$\flmx$ be the largest.  The $\proj$ map is required to satisfy the
following properties:
\begin{itemize}
\item if $r \notin [\flmi,\flmx]$, then $\proj (r) = \err$;
\item if $r = \inj (f)$ then $\proj (r) = f$.
\end{itemize}
We shall also later require $\proj$ to be measurable (see
Proposition~\ref{expressionsmeasurable}).

This can be achieved for example by the \emph{round-to-nearest}
function, defined by:
\[
\begin{array}{cccc}
\proj:&\br &\to &\bf\\
& r&\mapsto &\proj(r)=\left\{ \begin{array}{ll}\err & \text{if } r\notin [\flmi,\flmx]\\
\operatorname{argmin} \{|f-r|,\ f\in\ff\}& \text{otherwise} \end{array}\right.
\end{array}
\]
When $\operatorname{argmin} \{|f-r|,\ f\in\ff\}$ contains two
elements, the IEEE 754 standard specifies even rounding, i.e., we take
the value $f \in \ff$ whose ulp (last bit of the mantissa) is $0$.

\section{Concrete Semantics of 
  Expressions and Tests}
\label{sec:concrete}

We now construct two concrete semantics, the first one denoted by
$\semr{\cdot}$ on real numbers, the second one denoted by
$\semf{\cdot}$ on floating-point values.  The construction of these
semantics is based on the two maps $\inj$ and $\proj$ defined above.

\subsection{Concrete Semantics of Expressions}

Every expression will be interpreted in an \emph{environment} $\rho$,
which serves to specify the values of variables.  Simply, $\rho$ is a
map from the set $\Va$ of variables to $\br$ (in the real number
semantics) or to $\bf$ (in the floating-point semantics).  We denote
by $\Sigma_f$ the set of floating-point environments, and by
$\Sigma_r$ the set of real number environments.

We start with the semantics in the real model. Let $\rho_r$ be in
$\Sigma_r$.  The concrete semantics $\semr{expr}$ of expressions is
constructed in the  obvious way:

\[
\begin{array}{ccl}
\semr{a}(\rho_r)&=&a\\
\semr{x}(\rho_r)&=&\rho_r(x)\\
\semr{-e}(\rho_r)&=&-\semr{e}(\rho_r)\\
\semr{e_1\db{+}e_2}(\rho_r)&=&\semr{e_1}(\rho_r)+\semr{e_2}(\rho_r)\\
\semr{e_1\db{-}e_2}(\rho_r)&=&\semr{e_1}(\rho_r)-\semr{e_2}(\rho_r)\\
\semr{e_1\db{\times}e_2}(\rho_r)&=&\semr{e_1}(\rho_r)\times \semr{e_2}(\rho_r)\\
\semr{e_1\db{/}e_2}(\rho_r)&=& \semr{e_1}(\rho_r)/\semr{e_2}(\rho_r)
\end{array}
\]
The operations are well-defined by Convention~\ref{convention1}.

Now let us define the floating-point semantics. Let $\rho_f$ be in
$\Sigma_f$.  The floating-point semantics $\semf{expr}$ of expressions
is defined by rounding at the evaluation of each subexpression:

\[
\begin{array}{ccl}
\semf{a}(\rho_f)&=&\proj(a)\\
\semf{x}(\rho_f)&=&\rho_f(x)\\
\semf{-e}(\rho_f)&=&\proj \left(-\inj(\semf{e}(\rho_f)) \right)\\
\semf{e_1\db{+}e_2}(\rho_f)&=&\proj\left( \inj(\semf{e_1}(\rho_f))+\inj(\semf{e_2}(\rho_f)) \right)\\
\semf{e_1\db{-}e_2}(\rho_f)&=&\proj\left( \inj(\semf{e_1}(\rho_f))-\inj(\semf{e_2}(\rho_f)) \right)\\
\semf{e_1\db{\times}e_2}(\rho_f)&=&\proj\left( \inj(\semf{e_1}(\rho_f))\times \inj(\semf{e_2}(\rho_f)) \right)\\
\semf{e_1\db{/}e_2}(\rho_f)&=&\proj\left( \inj(\semf{e_1}(\rho_f))/\inj(\semf{e_2}(\rho_f)) \right)
\end{array}
\]

\subsection{Concrete Semantics of Tests}

The semantics of tests is a bit subtler.  Although one cannot
distinguish $\pinf$, $\minf$, $\NAN$ using expressions only---this
justified, at least partly, our decision to abstract them as a single
value $\err$---one can distinguish them using tests.  Experiments with
a C compiler (gcc 4.2.1 here) indeed show the following behaviors:
\[
\begin{array}{|cc|cccccc|}
  \hline
  a & b & a\verb/==/b & a\verb/!=/b & a\verb/<=/b & a\verb/</ b
  & a\verb/>=/ b & a \verb/>/ b \\
  \hline
  \pinf & \pinf & 1 & 0 & 1 & 0 & 1 & 0 \\
  \pinf & \minf & 0 & 1 & 0 & 0 & 1 & 1 \\
  \NAN & \NAN & 0 & 1 & 0 & 0 & 0 & 0 \\
  \hline
\end{array}
\]
Note for example that an $\NAN$ is not considered equal to itself,
that $a\verb/!=/ b$ is the negation of $a\verb/==/ b$ but $a\verb/>/
b$ is not the negation of $a\verb/<=/ b$ (e.g., when $a=b=\NAN$).

There are two ways we can deal with this phenomenon.  Either we
abandon the confusion of $\pinf$, $\minf$, $\NAN$ as the single value
$\err$, which will allow us to replay the above behavior precisely,
but will incur many complications; or we consider that the semantics
of tests must be \emph{non-deterministic}: not knowing whether $\err$
means $\pinf$, $\minf$, $\NAN$, we are forced to consider that $\err
\verb/==/ \err$ is \emph{any} value in $\{0, 1\}$.

So the semantics of tests will not be a single value, but a \emph{set}
of (Boolean, in $\{0, 1\}$) values.  One may say that our concrete
semantics is therefore slightly of an abstract semantics.  We count on
the fact that $\err$ abstracts (so-called silent) errors, and should
occur rarely in working programs.  (We are not after detecting subtle
errors, but to give reasonable accuracy bounds on actual working
programs.)

On the other hand, we do not need to specify which semantics,
floating-point or real, is meant: both will work in the same way for
tests.  Let us introduce the new notation $\seme{\cdot}$, where
$\star$ is either $f$ (floating-point) or $r$ (real).  We denote by
$\esigma$ the set of environment in this context. Let $\rhos$ be in
$\esigma$.

\[
\begin{array}{rcl}
  \seme{e_1\db{<=}e_2}(\rhos)&=&
  \left\{
    \begin{array}{ll}
      \{1\} &\text{if }\seme{e_1} (\rhos) \neq \err, \seme{e_2}
      (\rhos) \neq \err,
      \text{ and } \seme{e_1}(\rhos)\leq \seme{e_2}(\rhos) \\
      \{0\} &\text{if }\seme{e_1} (\rhos) \neq \err, \seme{e_2}
      (\rhos) \neq \err,
      \text{ and } \seme{e_1}(\rhos) > \seme{e_2}(\rhos) \\
      \{0, 1\} & \text{if } \seme{e_1} (\rhos) =\err \text{ or } \seme{e_2} (\rhos) =\err
    \end{array}
  \right.\\
  \seme{e_1\db{<}e_2}(\rhos)&=&
  \left\{
    \begin{array}{ll}
      \{1\} &\text{if }\seme{e_1} (\rhos) \neq \err, \seme{e_2}
      (\rhos) \neq \err,
      \text{ and } \seme{e_1}(\rhos) < \seme{e_2}(\rhos) \\
      \{0\} &\text{if }\seme{e_1} (\rhos) \neq \err, \seme{e_2}
      (\rhos) \neq \err,
      \text{ and } \seme{e_1}(\rhos) \geq \seme{e_2}(\rhos) \\
      \{0, 1\} & \text{if } \seme{e_1} (\rhos) =\err \text{ or } \seme{e_2} (\rhos) =\err
    \end{array}
  \right.\\
  \seme{e_1\db{==}e_2}(\rhos)&=&
  \left\{
    \begin{array}{ll}
      \{1\} &\text{if }\seme{e_1} (\rhos) \neq \err, \seme{e_2}
      (\rhos) \neq \err,
      \text{ and } \seme{e_1}(\rhos) = \seme{e_2}(\rhos) \\
      \{0\} &\text{if }\seme{e_1} (\rhos) \neq \err, \seme{e_2}
      (\rhos) \neq \err,
      \text{ and } \seme{e_1}(\rhos) \neq \seme{e_2}(\rhos) \\
      \{0, 1\} & \text{if } \seme{e_1} (\rhos) =\err \text{ or } \seme{e_2} (\rhos) =\err
    \end{array}
  \right.\\
  \seme{e_1\db{!=}e_2}(\rhos)&=&
  \left\{
    \begin{array}{ll}
      \{1\} &\text{if }\seme{e_1} (\rhos) \neq \err, \seme{e_2}
      (\rhos) \neq \err,
      \text{ and } \seme{e_1}(\rhos)\neq \seme{e_2}(\rhos) \\
      \{0\} &\text{if }\seme{e_1} (\rhos) \neq \err, \seme{e_2}
      (\rhos) \neq \err,
      \text{ and } \seme{e_1}(\rhos) = \seme{e_2}(\rhos) \\
      \{0, 1\} & \text{if } \seme{e_1} (\rhos) =\err \text{ or } \seme{e_2} (\rhos) =\err
    \end{array}
  \right.\\
  \seme{!t}(\rhos)&=&
  \{1-v \mid v \in \seme{t} (\rhos)\}
\end{array}
\]

The symbols $\leq$, $<$, $\geq$, $>$ in the right-hand sides above are
the usual relations on $\rr$.  So for example, the semantics of $e_1
\db{<=} e_2$ is well-defined because we only ever compare two elements
of $\real$, i.e., two elements of $\br$ other than $\err$.

\subsection{Measurability of Concrete Semantics of Expressions and Tests}

In the definition of the semantics, we will work with Lebesgue
integrals, or notions that generalize the Lebesgue integral.
It is well-known that one cannot posit that every function is
integrable without causing inconsistencies, and we shall therefore
have to check that every function that we integrate is
\emph{measurable}.

Measurability concerns are (mostly) irrelevant in the floating-point
semantics, if we remember that $\ff$, hence $\bf$, is finite, and that
every function between finite spaces is measurable.  But they are
definitely important in the real number semantics.

Measurability is defined relatively to specific $\sigma$-algebras.
The Borel $\sigma$-algebra on $\rr$---or, more generally, on any
topological space---is the smallest $\sigma$-algebra that contains all
open subsets.

We extend the topology of $\rr$ to one on $\re$ by extending the
standard metric on $\rr$ to the following:
\[
d(x,y)=\left\{
  \begin{array}{cr}
    +\infty &\text{ if } x=\err \text{ or } y=\err \text{ and } x\neq y\\
    0 & \text{ if } x=y=\err\\
    |x-y|& \text{ if } x,y\in \rr
  \end{array}
\right.
\]
The resulting topology has, as opens, all open subsets of $\rr$, the
singleton $\{\err\}$, and their unions.  This makes $\err$ an (the
unique) isolated point of $\re$.  Note that this topology is not the
topology of the classical one-point (Alexandroff) compactification of
$\rr$, in which a basis of open neighborhoods of $\err$ would be given
by the sets $(-\infty, a) \cup (b, +\infty) \cup \{\err\}$, and $\err$
would not be isolated.  The latter would also be a possible choice,
but would induce additional, irrelevant complications.

The subspace $\fe$ has the subspace topology: this is just the
discrete topology, since $\fe$ is finite.

We equip $\Sigma_r$ with the smallest topology that makes each map
$\rho \mapsto \rho (x)$ continuous, for each $x \in \Va$.  This makes
$\Sigma_r$ isomorphic to $\re^{\Va}$ with the product topology.

Similarly, we equip $\Sigma_f$ with the subspace topology from
$\Sigma_r$.  This is also the product topology on $\fe^{\Va}$, up to
isomorphism.  Note that this is not the discrete topology as soon as
$\Va$ is infinite: indeed, $\fe^{\Va}$ is compact and infinite in this
case, but all compact discrete topological spaces are finite.
This argument is however uselessly subtle: programs only use finitely
many variables anyway, and for $\Va$ finite, $\Sigma_f$ has the
discrete topology.

We write $\bor{\Sigma_{r}}$ and $\bor{\Sigma_{f}}$ for the
$\sigma$-algebras of Borel subsets of $\Sigma_{r}$ and $\Sigma_{f}$
respectively.  By standard results in topological measure theory (and
crucially using the fact that $\Va$ is countable), these are also the
product $\sigma$-algebras on the (measure-theoretic) product of $\Va$
copies of $\re$, resp.\ $\fe$.  (This is because $\re$, $\fe$ are
Polish spaces, and the Borel $\sigma$-algebra on a countable
topological product of Polish spaces coincides with the
$\sigma$-algebra of the measure-theoretic product of the spaces, each
with their Borel $\sigma$-algebra.)  This is a reassuring statement:
it states that we can harmlessly say ``product'' without having to say
whether this is a topological or measure-theoretic product.  There is
no such trap here.

A measurable map $f : X \to Y$ is one such that $f^{-1} (E)$ is a
Borel subset for every Borel subset $E$; it is equivalent to require
that, for every open subset $U$, $f^{-1} (U)$ is Borel.  In
particular, every continuous map is measurable.  When $Y$ is
second-countable, i.e., has certain so-called \emph{basic} opens such
that every open subset is the union of countably many basic opens,
then $f$ is measurable iff $f^{-1} (U)$ is Borel for every basic open
$U$.  We shall use this in proofs; in particular when $Y = \re$, where
we can take the intervals with rational endpoints, and $\{\err\}$, as
basic opens.

One might think that expressions have continuous real semantics, but
this is wrong: $x/y$ as a function of $x, y \in \re$ is not continuous
at any point of the form $(x, 0)$.  But they are measurable.  This
would be repaired if we had taken the topology of the 1-point
compactification of $\rr$ on $\re$, but we only need measurability.
On the other hand, we really need the topological and
measure-theoretic products to coincide, and while this would also be
true with the 1-point compactification, the argument would be slightly
more complex.

\begin{proposition}[Expressions are Measurable]
  \label{expressionsmeasurable}
  \begin{itemize}
  \item For every expression $e$, $\rho \mapsto \semr{e}(\rho)$ is a
    measurable function from $\Sigma_{r}$ to $\re$.
  \item if $\Va$ is finite or $\proj$ is measurable, then for every
    expression $e$, $\rho \mapsto \semf{e}(\rho)$ is a measurable
    function from $\Sigma_{f}$ to $\fe$.
  \end{itemize}
\end{proposition}
\begin{proof}
  We proceed by induction on expressions.  Let $a\in\qq$. The function
  $\rho \mapsto\semr{a}(\rho)$ is a constant function and thus it is
  continuous, hence measurable.  Let $x\in\Va$. The function $\rho
  \mapsto\semr{x}(\rho)$ is the coordinate projection on the $x$
  coordinate of $\rho$ and thus it is continuous, hence measurable.
  The case of expressions of the form $\db{-} e$, $e_1 \db{+} e_2$,
  $e_1 \db{-} e_2$, $e_1 \db{\times} e_2$ follows by induction
  hypothesis, using the fact that the corresponding operations on
  $\re$ are continuous.  To show this, it suffices to show that the
  inverse image of every basic open subset (i.e., open intervals of
  $\rr$, and $\{\err\}$) is open in $\Sigma_r$.  For example, the
  inverse image of an open subset of $\rr$ by $+$ is an open subset of
  $\rr \times \rr$, hence of $\re \times \re$, and the inverse image
  of the basic open subset $\{\err\}$ is $(\re \times \{\err\}) \cup
  (\{\err\} \times \re)$, hence open.  The case of $e_1 \db{/} e_2$ is
  slightly different as $/$ is not continuous on $\re \times \re$.
  But it is measurable, as we now show, by showing that the inverse
  image of any basic open subset is Borel.  The inverse image of any
  open interval of $\rr$ is open, since division is continuous at
  every point $(x, y)$ with $y \neq 0$.  And the inverse image of
  $\{\err\}$ by $/$ is the union of $\re \times \{\err\}$, of
  $\{\ell\} \times \re$, and of $\re \times \{0\}$.  The first two are
  open hence Borel, while the last one is the countable intersection
  $\bigcap_{n \geq 1} (\re \times (-\frac 1 n, \frac 1 n))$, hence is
  Borel.

  The second assertion is trivial if $\Va$ is finite, in which case
  all involved $\sigma$-algebras are discrete.  In the general case,
  it suffices to observe that $\inj$ and $\proj$ are measurable:
  $\inj$ is even continuous, since any function from a discrete space
  is, and the fact that $\proj$ is measurable is our assumption Using
  the fact that the composition of measurable functions is measurable,
  and using a similar induction as above, we conclude.  \qed
\end{proof}

All natural rounding functions $\proj$ are measurable, so the
assumptions we are making in Proposition~\ref{expressionsmeasurable}
will be satisfied.  E.g.,
\begin{lemma}
  \label{lemma:rn:mes}
  The round-to-nearest map, with even rounding, is measurable from
  $\re$ to $\fe$.
\end{lemma}
\begin{proof}
  Since the Borel $\sigma$-algebra on $\fe$ is discrete, it is enough
  to check that the inverse image of any single element $f \in \fe$ is
  Borel.

  If $f \in (\flmi, \flmx) \cap \ff$, and if the ulp of $f$ is $0$,
  then this inverse image is $[\frac {f+f'} 2, \frac {f+f''} 2]$
  ($(\frac {f+f'} 2, \frac {f+f''} 2)$ if the ulp of $f$ is not $0$),
  where $f'$ is the largest element of $\ff$ strictly less than $f$
  and $f''$ is the smallest element of $\ff$ strictly larger than $f$.

  If $f = \flmi$, then the inverse image of $f$ is $[\flmi, \frac
  {f+f''} 2]$ (if the ulp of $f$ is $0$; $[\flmi, \frac {f+f''} 2)$ if
  the ulp of $f$ is not $0$), where $f''$ is the smallest element of
  $\ff$ strictly larger than $f$.

  If $f = \flmx$, then the inverse image of $f$ is $[\frac {f+f'} 2,
  \flmx]$ (if the ulp of $f$ is $0$; $(\frac {f+f'} 2, \flmx]$ if the
  ulp of $f$ is not $0$), where $f'$ is the largest element of $\ff$
  strictly less than $f$.

  Finally, the inverse image of $\err$ is the union of $\{\err\}$, of
  $(-\infty, \flmi)$, and of $(\flmx, +\infty)$.

  All these sets are either open, or closed, and in any case Borel.  \qed
\end{proof}

Tests are interpreted as maps from $\esigma$ to $\pow^* \{0, 1\}$,
where $\pow^*$ denotes non-empty powerset, and are thus
multifunctions.  One of the standard notions of measurability for
multifunctions is to say that, given topological spaces $X$ and $Y$,
$f : X \to \pow^* (Y)$ is measurable if and only if $f^{-1} (\Diamond
U)$ is Borel for every \emph{open} subset $U$ of $Y$.  ($\Diamond U$
is the set of subsets that intersect $U$.)  If we understand $f$ as a
relation between elements of $X$ and elements of $Y$, this means that
the elements $x \in X$ that are related to some element of a given
open subset $U$ should be Borel.

\begin{proposition}[Tests are Measurable]
  \label{testsmeasurable}
  For every test $t$, $\rho\mapsto \semr{t}(\rho)$ is a measurable
  function from $\Sigma_r$ to $\pow^* \{0, 1\}$.  If $\Va$ is finite
  or $\proj$ is measurable, then $\rho\mapsto \semf{t}(\rho)$ is a
  measurable function from $\Sigma_f$ to $\pow^* \{0, 1\}$.
\end{proposition}
\begin{proof}
  It suffices to show that the inverse image of $\Diamond \{0\}$ and
  of $\Diamond \{1\}$ are Borel.  We proceed by induction on $t$.  Let
  $\star$ be either $f$ or $r$.

  If $t$ is of the form $e_1 \db{<=} e_2$, then $\seme{t} (\rhos)$
  contains $0$ if and only if $\seme{e_1 \db{-} e_2} (\rhos)$ is in
  $\{\err\} \cup (0, +\infty)$ (if $\star=r$; in $\inj^{-1} (\{\err\}
  \cup (0, +\infty))$ if $\star=f$).  The latter is open, and
  $\seme{e_1 \db{-} e_2}$ is measurable by
  Proposition~\ref{expressionsmeasurable}, so $\seme{t}^{-1} (\Diamond
  \{0\})$ is Borel.  Similarly, $\seme{t} (\rho)$ contains $1$ if and
  only if $\seme{e_1 \db{-} e_2} (\rho)$ is in $\{\err\} \cup
  (-\infty, 0]$ (if $\star=r$; its inverse image by $\inj$ if
  $\star=f$), which is closed, so $\seme{t}^{-1} (\Diamond \{1\})$ is
  Borel.  We proceed similarly if $t$ is of the form $e_1 \db{<} e_2$,
  $e_1 \db{==} e_2$, or $e_1 \db{!=} e_2$.

  FInally, if $t$ is of the form $!t'$, $\seme{t}^{-1} (\Diamond
  \{0\}) = \seme{t'}^{-1} (\Diamond \{1\})$, and $\seme{t}^{-1}
  (\Diamond \{1\}) = \seme{t'}^{-1} (\Diamond \{0\})$, which allows us
  to conclude immediately.  \qed
\end{proof}

\section{Weakest Preconditions and Continuation-Passing Style Semantics}
\label{sec:weakestconcret}

The idea of a continuation-passing style (CPS) semantics is that the
value $v$ returned by a given program is not given explicitly.
Rather, one passes a continuation parameter $\kappa$ to the semantics,
and the latter is defined so that it eventually calls $\kappa$ on the
final value $v$.

While this seems like a complicated and roundabout way of defining
semantics, this is very useful.  For example, this allows one to give
semantics to exceptions, or to various forms of non-determinism and
probabilistic choice \cite{Gou-csl07}.

The continuation $\kappa$ itself is a map from the domain of values to
some, usually unspecified domain of \emph{answers} $\ANS$.  (In
\cite{Gou-csl07}, $\ANS$ was required to be $\rr^+$.)

Also, the ``final value'' of a program should here be understood as
the final environment $\rhos$ that represents the state the program is
in on termination.  So a continuation $\kappa$ will be a map from
$\esigma$ to $\ANS$.

It should also be noted that continuation-passing style semantics are
nothing else than a natural generalization of Dijkstra's \emph{weakest
  preconditions}, or the computation of sets of predecessor states in
transition systems.  This is obtained by taking $\ANS = \{0, 1\}$.
Then the continuations $\kappa$ are merely the indicator maps of
subsets $E$ of environments (predicates $P$ on environments), and the
continuation-passing style denotation of program $\pi$ in continuation
$\kappa$ is merely the (continuation representing) the set of
environments $\rho$ such that evaluating $\pi$ starting from $\rho$
may terminate with an environment in $E$ (satisfying $P$).

Recall that an $\omega$-cpo is a poset in which every ascending
sequence $x_0 \leq x_1 \leq \ldots \leq x_n \leq \ldots$ has a
supremum (a least upper bound).
\begin{assum} 
  \label{assum:cpo}
  We assume that $\ANS$ is an $\omega$-cpo with a smallest element
  $\bot_{\ANS}$, and binary suprema.
\end{assum}
We write $\sup$ for suprema, and reserve $\sup^\uparrow$ for suprema
of ascending sequences.  Assumption~\ref{assum:cpo} can be stated
equivalently as: $\ANS$ has all countable suprema (including the
supremum $\bot_{\ANS}$ of the empty family).  If the language had been
deterministic (we fall short of this because of the way $\err$ is
dealt with in tests), we would only need $\ANS$ to be an $\omega$-cpo,
and would not have a need to binary suprema.

The typical example of such a set $\ANS$ of answers is $\rr_+ \cup
\{+\infty\}$, with its usual ordering.

As usual, we define the semantics of instructions by recursion on
syntax:
\begin{itemize}
\item $\skipp$:
  \[
  \weak{\lsupexp{\ell_1}{\skipp},\ell_2}(\kappa)=\kappa
  \]

\item assignment:
  \[
  \weak{\lsupexp{\ell_1}{x:=e},\ell_2}(\kappa)=\fun\ \rho\mapsto \kappa(\rho[x\to \seme{e}(\rho)])
  \]

\item sequence:
  \[
  \weak{\lsupexp{\ell_1}{P};\lsupexp{\ell_2}{Q},\ell_3}(\kappa)=
  \weak{\lsupexp{\ell_1}{P},\ell_2}\left( \weak{\lsupexp{\ell_2}{Q},\ell_3}(\kappa)\right)
  \]

\item tests:
  \[
  \begin{array}{l}
    \weak{\lsupexp{\ell}{\ifp}\ t\ \thenp \lsupexp{\ell_1}{P_1}\ \elsep \lsupexp{\ell_0}{P_0},\ell_2}(\kappa)=\\
    \\
    \fun\ \rho\mapsto
    \sup_{i \in \seme{t}(\rho)} \left(\weak{\lsupexp{\ell_i}{P_i},\ell_2}(\kappa)\right)(\rho)
  \end{array}
  \]
  In other words,
  \[
  \begin{array}{l}
    \weak{\lsupexp{\ell}{\ifp}\ t\ \thenp \lsupexp{\ell_1}{P_1}\ \elsep \lsupexp{\ell_0}{P_0},\ell_2}(\kappa)=\\
    \\
    \fun\ \rho\mapsto
    \left\{
      \begin{array}{lr}
        \left(\weak{\lsupexp{\ell_1}{P_1},\ell_2}(\kappa)\right)(\rho)&
        \text{if } \seme{t}(\rho)=\{1\}\\
        \left(\weak{\lsupexp{\ell_0}{P_0},\ell_2}(\kappa)\right)(\rho)&
        \text{if } \seme{t}(\rho)=\{0\}\\
        \sup\left(
          \left(\weak{\lsupexp{\ell_1}{P_1},\ell_2}(\kappa)\right)(\rho),
          \left(\weak{\lsupexp{\ell_0}{P_0},\ell_2}(\kappa)\right)(\rho)
        \right)
        & \text{if } \seme{t}(\rho) = \{0, 1\}
      \end{array}
    \right.
  \end{array}
  \]
\end{itemize}

The definition of the semantics of a loop, of the form
$\lsupexp{\ell_1}{\whilep}\ t\ \lsupexp{\ell_2}{P}$ uses an auxilary
map. We denote by $\mathbf{F}(\esigma,\ANS)$ the set
$\left(\left(\esigma\to \ANS\right)\to \left(\esigma\to
    \ANS\right)\right)$ i.e. the set of maps from $\left(\esigma\to \ANS\right)$ to itself. We equip $\left(\esigma\to \ANS\right)$ with the pointwise 
    ordering. The set $\mathbf{F}(\esigma,\ANS)$ is also equipped with
the pointwise ordering: $f\leq g$ iff for every
$\kappa\in\left(\esigma\to \ANS\right)$, for every $\rho\in\esigma$,
$f(\kappa)(\rho)\leq g(\kappa)(\rho)$ in $\ANS$.  For every countable
family $\left(f_{i}\right)_{i\in I}$ of elements of
$\mathbf{F}(\esigma,\ANS)$, its supremum $\sup_{i\in I} f_{i}$ is then
also computed pointwise:
\[
\sup_{i\in I} f_{i}:\kappa\mapsto \left(\fun \rho\mapsto \sup_{i\in I} \left( f_{i}(\kappa)(\rho) \right)\right)\enspace.
\]

From this latter definition, we get the following lemma.

\begin{lemma}
The set $\mathbf{F}(\esigma,\ANS)$ is a $\omega$-cpo with binary suprema,
and with a smallest element $\bot_{\mathbf{F}(\esigma,\ANS)}$ defined
as:
\[
\bot_{\mathbf{F}(\esigma,\ANS)}(\kappa)=\fun \rho \mapsto\bot_{\ANS},\ \forall\, \kappa:\esigma\mapsto \ANS\enspace .
\]
\end{lemma}

\begin{itemize}
\item loops.  Given a test $t$ and an instruction
  $\lsupexp{\ell_2}{P}, \ell_1$, let $H_{t, \lsupexp{\ell_2}{P},
    \ell_1}$ be the map from $\mathbf{F}(\esigma,\ANS)$ to
  $\mathbf{F}(\esigma,\ANS)$ defined as follows:
  \[
  \left(H_{t, \lsupexp{\ell_2}{P}, \ell_1}
    (\varphi)\right)(\kappa)(\rho)= \left\{
    \begin{array}{ll}
      \left(\weak{\lsupexp{\ell_2}{P},\ell_1}\left(\varphi(\kappa)\right)\right)(\rho)
      & \text{if } \seme{t}(\rho)=\{1\}\\
      \kappa(\rho)& \text{if } \seme{t}(\rho)=\{0\}\\
      \sup
      (\left(\weak{\lsupexp{\ell_2}{P},\ell_1}\left(\varphi(\kappa)\right)\right)(\rho),
      \kappa(\rho))
      & \text{if } \seme{t}(\rho) = \{0, 1\}
    \end{array}
  \right.
  \]
  So, for example, $\weak{\lsupexp{\ell}{\ifp}\ t\ \thenp
    \lsupexp{\ell_1}{P_1}\ \elsep \lsupexp{\ell_0}{P_0},\ell_2}=H_{t,
    \lsupexp{\ell_1}{P_1}, \ell_2} (\weak{\lsupexp{\ell_0}{P_0},
    \ell_2})$.

  The semantics of the loop $\lsupexp{\ell_1}{\whilep}\ t\
  \lsupexp{\ell_2}{P}, \ell_3$ is the supremum of the sequence
  $\bot_{\mathbf{F}(\esigma,\ANS)}$, $H_{t, \lsupexp{\ell_2}{P},
    \ell_3} (\bot_{\mathbf{F}(\esigma,\ANS)})$, $H_{t,
    \lsupexp{\ell_2}{P}, \ell_3} (H_{t, \lsupexp{\ell_2}{P}, \ell_3}
  (\bot_{\mathbf{F}(\esigma,\ANS)}))$, \ldots in
  $\mathbf{F}(\esigma,\ANS)$, namely:
  \[
  \weak{\lsupexp{\ell_1}{\whilep}\ t \lsupexp{\ell_2}{P},\ell_3}=
  \sup_{n \in \nn} H_{t, \lsupexp{\ell_2}{P}, \ell_3}^n
  (\bot_{\mathbf{F}(\esigma,\ANS)})
  \]
\end{itemize}
A more standard definition would have been to let
$\weak{\lsupexp{\ell_1}{\whilep}\ t \lsupexp{\ell_2}{P},\ell_3}$ be
defined as the least fixpoint of $H_{t, \lsupexp{\ell_2}{P}, \ell_3}$
in $\mathbf{F}(\esigma,\ANS)$.  We show below that this would be
equivalent.  The reason is that the map $H_{t, \lsupexp{\ell_2}{P},
  \ell_3}$ is $\omega$-Scott-continuous, i.e., is monotone and
preserves suprema of ascending sequences.

We prove this through two lemmas.  The first one shows that $H_{t, \lsupexp{\ell_2}{P}, \ell_1}$ is
$\omega$-Scott-continuous when the maps
$\kappa\mapsto\weak{\lsupexp{\ell_2}{P},\ell_{1}}(\kappa)$ are
$\omega$-Scott-continuous. This second lemma says that the maps
$\kappa\mapsto\weak{\lsupexp{\ell_2}{P},\ell_{1}}(\kappa)$ are
actually $\omega$-Scott-continuous.

\begin{lemma}
  \label{wpprop}
  Let $\lsupexp{\ell_{2}}{P}$ be an instruction. Let $t$ be a test.
  Assume that the map
  \begin{equation}
    \label{wpScott}
    \kappa\mapsto\weak{\lsupexp{\ell_2}{P},\ell_{1}}(\kappa) \text{ is $\omega$-Scott-continuous}\enspace ,
  \end{equation}
  then:
  \begin{itemize}
  \item The map $H_{t, \lsupexp{\ell_2}{P}, \ell_1}$ is $\omega$-Scott-continuous.
  \item The map $\sup_{n \in \nn} H_{t, \lsupexp{\ell_2}{P},
      \ell_1}^n$ is $\omega$-Scott-continuous.
  \item $\kappa\mapsto\weak{\lsupexp{\ell_1}{\whilep}\ t
      \lsupexp{\ell_2}{P},\ell_3}(\kappa)$ is
    $\omega$-Scott-continuous.
  \end{itemize}
\end{lemma}
\begin{proof}
  Let us prove the first assertion.  Let
  $\varphi,\varphi'\in\mathbf{F}(\esigma,\ANS)$ such that $\varphi\leq
  \varphi'$. For every $\kappa:\esigma\mapsto \ANS$, for every $\rho$
  such that $\seme{t}(\rho)=\{1\}$,
  \[
  \weak{\lsupexp{\ell_2}{P},\ell_1}(\varphi(\kappa))(\rho)\leq \weak{\lsupexp{\ell_2}{P},\ell_1}(\varphi'(\kappa))(\rho)\enspace,
  \]
  so:
  \[
  H_{t,
    \lsupexp{\ell_2}{P}, \ell_1}(\varphi)(\kappa)(\rho)\leq H_{t,
    \lsupexp{\ell_2}{P}, \ell_1}(\varphi')(\kappa)(\rho)\enspace.
  \]
  Let $\left(\varphi_{n}\right)_{n\in\nn}$ be an ascending sequence in
  $\mathbf{F}\left(\esigma,\ANS\right)$. We also have:
  \[
  \left(\weak{\lsupexp{\ell_2}{P},\ell_1}\left(\supi_{n\in\nn}\varphi_{n}(\kappa)\right)\right)(\rho)=
  \supi_{n\in\nn}\left(\weak{\lsupexp{\ell_2}{P},\ell_1}\left(\varphi_{n}(\kappa)\right)\right)(\rho)\enspace. 
  \]
  When $\seme{t}(\rho)=\{1\}$, this is equivalent to:
  \[
  H_{t,
    \lsupexp{\ell_2}{P}, \ell_1}\left(\supi_{n\in\nn}\varphi_{n}\right)(\kappa)(\rho)=\left(\supi_{n\in\nn}H_{t,
      \lsupexp{\ell_2}{P}, \ell_1}(\varphi_{n})\right)(\kappa)(\rho)\enspace.
  \]
  When $\seme{t}(\rho)=\{0\}$, we obtain the same equality, where now
  both sides are the constant $\kappa(\rho)$.  When
  $\seme{t}(\rho)=\{0, 1\}$,
  \begin{eqnarray*}
    H_{t,
      \lsupexp{\ell_2}{P},
      \ell_1}\left(\supi_{n\in\nn}\varphi_{n}\right)(\kappa)(\rho)
    & = &
    \sup
    \left(\supi_{n\in\nn}\left(\weak{\lsupexp{\ell_2}{P},\ell_1}\left(\varphi_{n}(\kappa)\right)\right)(\rho),
      \kappa (\rho)
    \right)\\
    & = & 
    \supi_{n\in\nn}\left(
      \sup
      \left(\weak{\lsupexp{\ell_2}{P},\ell_1}\left(\varphi_{n}(\kappa)\right)(\rho),
        \kappa (\rho)
      \right)
    \right)\\
    & = & \left(\supi_{n\in\nn}H_{t,
      \lsupexp{\ell_2}{P}, \ell_1}(\varphi_{n})\right)(\kappa)(\rho)\enspace.
  \end{eqnarray*}

  The second assertion follows from the first, from the fact that
  compositions of $\omega$-Scott-continuous maps are again
  $\omega$-Scott-continuous (hence $H_{t, \lsupexp{\ell_2}{P},
    \ell_1}^n$ is $\omega$-Scott-continuous for every $n \in \nn$),
  and that suprema of $\omega$-Scott-continuous are
  $\omega$-Scott-continuous.

  The last assertion follows trivially from the second one, using the
  fact that application (of maps to $\bot_{\mathbf{F}(\esigma,\ANS)}$)
  is $\omega$-Scott-continuous.  \qed
\end{proof}

Next, we prove that for every instruction $\lsupexp{\ell}{P}$, for
every label $\ell'$, the map
$\kappa\mapsto\weak{\lsupexp{\ell}{P},\ell'}(\kappa)$ is
$\omega$-Scott-continuous.

\begin{lemma}[$\omega$-Scott-Continuity of $\weak{\cdot}$]
  \label{weakScott}
  For every instruction $\lsupexp{\ell}{P}$, for every label $\ell'$,
  the map $\kappa\mapsto\weak{\lsupexp{\ell}{P},\ell'}(\kappa)$ is
  $\omega$-Scott-continuous.
\end{lemma}

\begin{proof}
  We proceed by induction on the instructions. 

  $\bullet$ $\skipp$. The instruction $\skipp$ is the identity map
  from $\esigma\to \ANS$ to itself, so it is
  $\omega$-Scott-continuous.

  $\bullet$ Assignment. Let $\kappa,\kappa'$ be maps from $\esigma$ to
  $\ANS$ such that $\kappa\leq \kappa'$. For every $\rho\in\esigma$,
  $\kappa(\rho[x \to\seme{e}(\rho)])\leq \kappa'(\rho[x
  \to\seme{e}(\rho)])$.  So $\weak{\lsupexp{\ell_1}{x:=e},\ell_2}$ is
  a monotonic map. Now, we consider an ascending sequence
  $(\kappa_n)_{n\in\nn}$ of maps from $\esigma$ to $\ANS$. We have:
  \[
  \left(\supi_{n\in\nn}(\kappa_n)\right)(\rho[x \to\seme{e}(\rho)])
  =\supi_{n\in\nn}\left(\kappa_n(\rho[x \to\seme{e}(\rho)])\right)
  \]
  and then:
  \[
  \weak{\lsupexp{\ell_1}{x:=e},\ell_2}\left(\supi_{n\in\nn}\kappa_n\right)=\supi_{n\in\nn} 
  \weak{\lsupexp{\ell_1}{x:=e},\ell_2}(\kappa_n)
  \enspace.
  \]

  $\bullet$ Sequence.  By induction hypothesis on
  $\lsupexp{\ell_{1}}{P}$ and $\lsupexp{\ell_{2}}{Q}$, the maps
  $\kappa \to \weak{\lsupexp{\ell_{1}}{P},\ell_{2}}(\kappa)$ and
  $\kappa \to \weak{\lsupexp{\ell_{2}}{Q},\ell_{3}}(\kappa)$ are
  $\omega$-Scott-continuous.  Since the composition of two
  $\omega$-Scott-continuous maps is $\omega$-Scott-continuous then the
  sequence $\kappa \to
  \weak{\lsupexp{\ell_{1}}{P};\lsupexp{\ell_{2}}{Q},\ell_{3}}(\kappa)$
  is also $\omega$-Scott-continuous.

  $\bullet$ Tests.  By induction hypothesis on $\lsupexp{\ell_{2}}{P}$
  and $\lsupexp{\ell_{3}}{Q}$, the maps $\kappa \to
  \weak{\lsupexp{\ell_{2}}{P},\ell_{4}}(\kappa)$ and $\kappa \to
  \weak{\lsupexp{\ell_{3}}{Q},\ell_{4}}(\kappa)$ are
  $\omega$-Scott-continuous.  Since we consider a pointwise order, it
  suffices to show that for every $\rho\in\esigma$,
  $\kappa\mapsto\weak{\lsupexp{\ell_1}{\ifp}\ t\ \thenp
    \lsupexp{\ell_2}{P}\ \elsep
    \lsupexp{\ell_3}{Q},\ell_4}(\kappa)(\rho)$ is
  $\omega$-Scott-continuous (from $\esigma\mapsto \ANS$ to
  $\ANS$). When we fix $\rho\in \esigma$, we get three cases whether the
  test is true, false or true and false.  In each case, we conclude that
  $\weak{\lsupexp{\ell_1}{\ifp}\ t\ \thenp \lsupexp{\ell_2}{P}\ \elsep
    \lsupexp{\ell_3}{Q},\ell_4}(\kappa)$ is $\omega$-Scott-continuous
  by induction hypothesis.

  $\bullet$ Loops.  By induction hypothesis, $\kappa \to
  \weak{\lsupexp{\ell_{2}}{P},\ell_{1}}(\kappa)$ is
  $\omega$-Scott-continuous. Lemma~\ref{wpprop} immediately entails
  that $\kappa\mapsto \weak{\lsupexp{\ell_1}{\whilep}\ t
    \lsupexp{\ell_2}{P},\ell_3}(\kappa)$ is
  $\omega$-Scott-continuous.  \qed
\end{proof}

We introduce a parametric version of classical previsions. Since we work 
with $\omega$-cpos, we have to consider $[0,+\infty]$, we add arithmetics 
conventions to deal with $+\infty$. 

\begin{convention}[Arithmetics in $\rr_+\cup\{+\infty\}$]
We add the following rules:
\begin{itemize}
\item $0\times (+\infty)=(+\infty)\times 0=0$;
\item $+\infty\times +\infty=+\infty$;
\item For all $x\in [0,+\infty]$, $x+(+\infty)=(+\infty)+x=+\infty$.
\end{itemize}
\end{convention}

Let $X$ be a topological space. We equip $X$ with its Borel $\sigma$-algebra.
We denote by $\mesp{X}$ the set of positive measurable functions on $X$.

\begin{defi}[Parametric prevision]
\label{def:parametric}
Let $X$ be a non-empty set. Let $F$ be a map from $\mesp{X}$ to itself.
The map $F$ is said to be a parametric prevision if:
\begin{enumerate}
\item $F$ is positively homogeneous;
\item $F$ is monotonic;
\end{enumerate}
Moreover, a parametric prevision is said to be:
\begin{enumerate}
\item (lower) $F(f+g)\geq F(f)+F(g)$, for all functions $f,g\in\mesp{X}$;
\item (upper) $F(f+g)\leq F(f)+F(g)$, for all functions $f,g\in\mesp{X}$;
\item (linear) $F(f+g)=F(f)+F(g)$, for all functions $f,g$ from $\mesp{X}$;
\item ($\omega$-continuous) for all ascending family $(f_n)_{n\in \nn}$, $F(\supi_{n\in \nn} f_n)=\supi_{n\in\nn} F(f_n)$.
\end{enumerate}
\end{defi}

We recall that the set of positive measurable functions is a convex cone stable by countable infima and suprema 
and pointwise limit. The set of positive measurable functions contains constant (positive), and continuous functions.

\begin{proposition}
\label{compo}
\begin{enumerate}
\item The set of upper $\omega$-continuous parametric prevision is $\omega$-cpo (equipped with the pointwise ordering) 
with a smallest element (the null pfunctional $\underline{0}$ associates 
at $f\in\mesp{X}$ the positive measurable function $g:x\to 0$)
\item The set of upper $\omega$-continuous parametric prevision is stable by binary suprema. 
\item The set of upper $\omega$-continuous parametric prevision is stable by composition.
\end{enumerate}
\end{proposition}

\begin{proof}
\begin{enumerate}
\item The parametric $\underline{0}$ is clearly an upper $\omega$-continuous parametric prevision. Since 
the null parametric is the smallest element of $\mathbf{F}(X,[0,+\infty])$, it is also 
the smallest element of the set of upper $\omega$-continuous parametric previsions.

Let $\left(F_n\right)_{n\in\nn}$ be an ascending sequence. Since $\mesp{X}$ is stable by countable suprema, then
$(\supi_{n\in \nn} F_n) (f)=\supi_{n\in \nn} (F_n (f))\in\mesp{X}$ for all $f\in\mesp{X}$.

Let $\alpha\geq 0$ and $f\in\mesp{X}$. The set $\mesp{X}$ is a cone thus $\alpha f\in\mesp{X}$. 
Since $F_n$ are positively homogeneous then $(\supi_{n\in \nn} F_n) (\alpha f)=\supi_{n\in \nn} 
(F_n(\alpha f))=\supi_{n\in \nn}\alpha F_n(f)$ and since $\alpha\geq 0$, we conclude that 
$\supi_{n\in \nn} (F_n(\alpha f))=\alpha \supi_{n\in \nn} F_n(f)$ and $\supi_{n\in \nn} F_n$ 
is positively homogeneous. 

Let $f,g\in\mesp{X}$ such that $f\leq g$, $(\supi_{n\in \nn} F_n)(f)= \supi_{n\in \nn} F_n(f)$.
For all $n\in\nn$, $F_n(f)\leq F_n(g)$ and we get $\supi_{n\in \nn} F_n(f)\leq \supi_{n\in \nn} F_n(g)
=(\supi_{n\in \nn} F_n)(g)$ and we conclude that $\supi_{n\in \nn} F_n$ is monotonic. 

Let $f,g\in\mesp{X}$. For all $n\in\nn$, we have
$F_n(f+g)\leq F_n(f)+F_n(g)$, taking the suprema we get $\supi_{n\in \nn} F_n(f+g)\leq 
\supi_{n\in \nn}(F_n(f)+F_n(g))\leq \supi_{n\in \nn}F_n(f)+\supi_{n\in \nn}F_n(g)$ and $\supi_{n\in \nn}F_n$ 
is an upper parametric prevision. 

Now let $(f_k)_{k\in \nn}$ be an ascending sequence of elements of $\mesp{X}$. The set $\mesp{X}$ 
is stable by suprema hence $\supi_{k\in \nn} f_k\in\mesp{X}$. We have $\supi_{n\in \nn} F_n(\supi_{k\in \nn} 
f_k)=\supi_{n\in \nn} \supi_{k\in \nn}F_n(f_k)$ and the suprema commute and then $\supi_{n\in \nn} F_n(\supi_{k\in \nn}
f_k)=\supi_{k\in \nn} \supi_{n\in \nn}F_n(f_k)$. We conclude that $\supi_{n\in \nn} F_n$ is $\omega$-continuous.

\item Let $F,G$ be two upper $\omega$-continuous parametric prevision. From the supremum stability property,
$\sup(F,G)(f)=\sup(F(f),G(f))$ belongs to $\mesp{X}$ for all $f\in\mesp{X}$.

Let $\alpha\geq 0$ and $f\in\mesp{X}$ ($\alpha f\in\mesp{X}$), since $F,G$ are positively homogeneous then 
$\sup(F,G) (\alpha f)=\sup(F(\alpha f),G(\alpha f))=\sup(\alpha F(f),\alpha G(f))$ and since $\alpha\geq 0$, 
we conclude that $\sup(F,G) (\alpha f)=\alpha\sup(F(f),G(f))=\alpha\sup(F,G)(f)$ and $\sup(F,G)$ is positively 
homogeneous. 

The supremum of monotonic function is a monotonic function hence $\sup(F,G)$ is monotonic.

Let $f,g\in\mesp{X}$. We have $F(f+g)\leq F(f)+F(g)$ and $G(f+g)\leq G(f)+G(g)$, taking the supremum we get 
$\sup(F,G)(f+g)\leq \sup(F(f)+F(g),G(f)+G(g))\leq \sup(F(f),G(f))+\sup(F(g),G(g))=\sup(F,G)(f)+\sup(F,G)(g)$ 
and $\sup(F,G)$ is an upper parametric prevision. 

Now let $(f_k)_{k\in nn}$ be an ascending sequence of elements of $\mesp{X}$ (and thus $\supi_{k\in \nn} f_k\in\mesp{X}$.
We have $\sup(F,G)(\supi_{k\in \nn}f_k)=\sup(F(\supi_{k\in \nn} f_k),G(\supi_{k\in \nn} f_k))=\sup(\supi_{k\in \nn} 
F(f_k),\supi_{k\in \nn} G(f_k))$ and the suprema commute and then $\sup(F,G)(\supi_{k\in \nn} f_k)=\supi_{k\in \nn} 
\sup(F,G)(f_k)$. We conclude that $\sup(F,G)$ is $\omega$-continuous.

\item Let $F,G$ be two upper $\omega$-continuous parametric prevision. 

Since for all $f,g\in\mesp{X}$, $F(f)$ and $G(g)$ belong to $\mesp{X}$ then 
for all $h\in\mesp{X}$, $F(G(h))$ belongs to $\mesp{X}$.

Let $\alpha\geq 0$ and $f\in\mesp{X}$, since $F,G$ are positively homogeneous then 
$F\circ G (\alpha f)=F(G(\alpha f))=F(\alpha G(f))=\alpha F\circ G(f)$, 
we conclude that  $F\circ G$ is positively homogeneous.

The composition of two monotonic maps is also monotonic thus $F\circ G$ is monotonic.

Let $f,g\in\mesp{X}$. We have $G(f+g)\leq G(f)+G(g)$, and since $F$ is monotonic,
$F\circ G(f+g)\leq F(G(f)+G(g))\leq F\circ G(f)+F\circ G(g)$ and $F\circ G$ is an upper 
parametric prevision. 

Now let $(f_k)_{k\in nn}$ be an ascending sequence in $\mesp{X}$. We have 
$F\circ G(\supi_{k\in \nn} f_k)=F(G(\supi_{k\in \nn} f_k))
=F(\supi_{k\in \nn} G(f_k))=\supi_{k\in \nn} F\circ G(f_k)$ and 
we conclude that $F\circ G$ is $\omega$-continuous.
\end{enumerate}
\end{proof}

The main difference between prevision and parametric prevision is the co-domain.
Since the domain and the co-domain are the same, we can compose two parametric previsions 
to construct a new one. It allows us to think about least fixed points of parametric previsions.

\begin{defi}[Previsions]
Let $X$ be a non-empty set. Let $F$ be a map from $\mesp{X}$ to $[0,+\infty]$.
The map $F$ is said to be a prevision if:
\begin{enumerate}
\item $F$ is positively homogeneous;
\item $F$ is monotonic;
\end{enumerate}
Moreover, a prevision is said to be:
\begin{enumerate}
\item (lower) $F(f+g)\geq F(f)+F(g)$, for all functions $f,g\in\mesp{X}$;
\item (upper) $F(f+g)\leq F(f)+F(g)$, for all functions $f,g\in\mesp{X}$;
\item (linear) $F(f+g)=F(f)+F(g)$, for all functions $f,g\in\mesp{X}$;
\item ($\omega$-continuous) for all ascending family $(f_n)_{n\in \nn}\in\mesp{X}$, 
$F(\supi_{n\in \nn} f_n)=\supi_{n\in\nn} F(f_n)$.
\end{enumerate}
\end{defi}

The set $\mesp{X}$ is equipped with the pointwise ordering. The following proposition shows why the 
term \emph{parametric} appears in Definition~\ref{def:parametric}. The space of parameters is the same 
of domain of the functions of $\mesp{X}$ i.e. the set X. When we fix a parameter, we get a classical 
prevision.

\begin{proposition}[parametric prevision and previsions]
\label{equiv}
The parametric $F$ from $\mesp{X}$ to itself is a parametric (upper,lower,linear,$\omega$-continuous) 
prevision iff for all $x\in X$, the map $F_x$ from $\mesp{X}$ to $[0,+\infty]$ defined as $F_x(f)=F(f)(x)$ 
for all $f\in\mesp{X}$ is a (upper,lower,linear,$\omega$-continuous) classical prevision and the maps 
$x\to F_x(h)$ are measurable for all $h\in\mesp{X}$.
\end{proposition}

The nondeterminism due to the tests and the value $\err$ implies that we cannot expect linearity. Indeed the 
binary supremum of the sum is not equal to the sum of the suprema, we have only an inequality. In the case of 
$Ans=\rr_+\cup\{+\infty\}$, we can establish that the weakest preconditions and continuation-passing style 
semantics defines an upper parametric prevision. To prove this result, we need a lemma which says that the 
semantics maps $\mesp{X}$ to itself.

\begin{lemma}
\label{weakmeasurable}
If $\kappa\in\mesp{\esigma}$, then, for all instructions $\lsupexp{\ell_{1}}{P}$, 
$\weak{\lsupexp{\ell_1}{P},\ell_{2}}(\kappa)\in\mesp{\esigma}$.
\end{lemma}

\begin{proof}
We prove this result by induction on the instructions. 

$\bullet$ Since the $\skipp$ is the identity map, thus $\kappa\in\mesp{\esigma}$ implies that
$\weak{\lsupexp{\ell_1}{\skipp}}(\kappa)$ also belongs to $\mesp{\esigma}$. 

$\bullet$ Now, we consider the assignment. We define the map $h:\esigma\mapsto\esigma$ such that 
at $\rho$ $h$ associates $\rho(y)$ if $y\neq x$ and $\seme{e}(\rho)$ otherwise.
A coordinate of $h$ is either coordinate projection or the concrete semantics 
of an expression which from Proposition~\ref{expressionsmeasurable} is measurable. 
We conclude that $h$ is measurable since it is componentwise measurable.
We conclude that $\weak{\lsupexp{\ell_1}{x:=e},\ell_2}(\kappa)\in\mesp{\esigma}$ 
(composition of) for all $\kappa\in\mesp{\esigma}$. 

$\bullet$  Let $\kappa\in\mesp{\esigma}$. The function 
$\weak{\lsupexp{\ell_1}{P};\lsupexp{\ell_2}{Q},\ell_{3}}(\kappa)$ is defined as 
$\weak{\lsupexp{\ell_1}{P};\ell_{2}}(\weak{\lsupexp{\ell_2}{Q},\ell_{3}}(\kappa))$. 
Suppose that $\weak{\lsupexp{\ell_1}{P},\ell_{2}}(\kappa')$ and $\weak{\lsupexp{\ell_2}{Q},\ell_{3}}(\kappa'')$ belong
to $\mesp{\esigma}$ for all $\kappa',\kappa''\in\mesp{\esigma}$. Then $\kappa':=\weak{\lsupexp{\ell_2}{Q},\ell_{3}}(\kappa))$ is positive and measurable.
We conclude that $\weak{\lsupexp{\ell_1}{P};\ell_{2}}(\kappa')\in\mesp{\esigma}$.

$\bullet$ Let $\kappa\in\mesp{\esigma}$. We have:
\[
\begin{array}{ll}
\weak{\lsupexp{\ell_1}{\ifp}\ t\ \thenp \lsupexp{\ell_2}{P}\ \elsep \lsupexp{\ell_3}{Q},\ell_4}(\kappa){}
&=\sup\left(\weak{\lsupexp{\ell_1}{P},\ell_{4}}(\kappa),\weak{\lsupexp{\ell_2}{Q},\ell_{4}}(\kappa)\right)
\chi_{\{\rho\mid\seme{t}(\rho)=\{0,1\}\}}\\
&+\weak{\lsupexp{\ell_1}{P},\ell_{4}}(\kappa)\chi_{\{\rho\mid\seme{t}(\rho)=\{1\}\}}\\
&+\weak{\lsupexp{\ell_2}{Q},\ell_{4}}(\kappa)\chi_{\{\rho\mid\seme{t}(\rho)=\{0\}\}}
\end{array}
\]

Suppose that $\weak{\lsupexp{\ell_1}{P},\ell_{4}}(\kappa)$ and $\weak{\lsupexp{\ell_2}{Q},\ell_{4}}(\kappa)$ 
are in $\mesp{\esigma}$. From Proposition~\ref{testsmeasurable} the functions $\chi_{\{\rho\mid\seme{t}(\rho)={1}\}}$,
$\chi_{\{\rho\mid\seme{t}(\rho)={0}\}}$ and $\chi_{\{\rho\mid\seme{t}(\rho)=\{0,1\}\}}$ are positive measurable functions. 
Since $\mesp{\esigma}$ is stable by product, sum and binary suprema, thus 
$\weak{\lsupexp{\ell_1}{\ifp}\ t\ \thenp \lsupexp{\ell_2}{P}\ \elsep \lsupexp{\ell_3}{Q},\ell_4}(\kappa)\in\mesp{\esigma}$. 

$\bullet$ Let $\kappa\in\mesp{\esigma}$. We have, from Lemma~\ref{wpprop}:
\[
\begin{array}{ll}
\weak{\lsupexp{\ell_1}{\whilep}\ t \lsupexp{\ell_2}{P},\ell_3}(\kappa)=\left(\lfp{H_{t,
      \lsupexp{\ell_2}{P}, \ell_1}}\right)(\kappa)
      &=\left(\supi_{n\in\nn} H_{t,\lsupexp{\ell_2}{P}, \ell_1}^{n}(\bot_{\mathbf{F}(\esigma,\brp)})\right)(\kappa)\\
      &=\supi_{n\in\nn} \left(H_{t,\lsupexp{\ell_2}{P}, \ell_1}^{n}(\bot_{\mathbf{F}(\esigma,\brp)})(\kappa)\right)
\end{array}
\]
We suppose that
$\weak{\lsupexp{\ell_2}{P},\ell_1}(\kappa')\in\mesp{\esigma}$ for all
$\kappa'\in\mesp{\esigma}$.Since $\mesp{\esigma}$ is an
$\omega$-cpo the smallest of which is the null function.  It suffices
to show that for all $n\in\nn$,
$\left(H_{t,
      \lsupexp{\ell_2}{P}, \ell_1}^{n}(\bot_{\mathbf{F}(\esigma,\brp)})\right)(\kappa)$
belongs to $\mesp{\esigma}$.  We prove this property by induction on
integers. The null function is positive and measurable. Now, we
suppose that there exists an integer $n$ such that
$\left(H_{t,
      \lsupexp{\ell_2}{P}, \ell_1}^{n}(\bot_{\mathbf{F}(\esigma,\brp)})\right)(\kappa)$
belongs to $\mesp{\esigma}$. We have:
\[
\begin{array}{ll}
\left(H_{t,\lsupexp{\ell_2}{P}, \ell_1}^{n+1}(\bot_{\mathbf{F}(\esigma,\brp)})\right)(\kappa)
&=\sup\left(\weak{\lsupexp{\ell_2}{P},\ell_1}\left(\left(H_{t,\lsupexp{\ell_2}{P}, \ell_1}^{n}
(\bot_{\mathbf{F}(\esigma,\brp)})\right)(\kappa)\right),\kappa\right)\chi_{\{\rho\mid\seme{t}(\rho)=\{0,1\}\}}\\
&+\weak{\lsupexp{\ell_2}{P},\ell_1}\left(\left(H_{t,\lsupexp{\ell_2}{P}, \ell_1}^{n}
(\bot_{\mathbf{F}(\esigma,\brp)})\right)(\kappa)\right)\chi_{\{\rho\mid\seme{t}(\rho)=\{1\}\}}\\
&+\kappa\chi_{\{\rho\mid\seme{t}(\rho)=\{0\}\}}
\end{array}
\]
From induction hypothesis (on instructions and integers $n$), Proposition~\ref{testsmeasurable} and stability of product, sum in $\mesp{\esigma}$
and binary suprema, we conclude that $\left(H_{t,
      \lsupexp{\ell_2}{P}, \ell_1}^{n+1}(\bot_{\mathbf{F}(\esigma,\brp)})\right)(\kappa)$ belongs to $\mesp{\esigma}$.
In conclusion, for all $n\in\nn$, $\left(H_{t,
      \lsupexp{\ell_2}{P}, \ell_1}^{n}(\bot_{\mathbf{F}(\esigma,\brp)})\right)(\kappa)$ belongs to $\mesp{\esigma}$ and 
$\weak{\lsupexp{\ell_1}{\whilep}\ t \lsupexp{\ell_2}{P},\ell_3}(\kappa)\in\mesp{\esigma}$.
\end{proof}

\begin{proposition}
When $Ans=\rr_+\cup\{+\infty\}$ and $X=\esigma$, for every instruction $\lsupexp{\ell}{P}$, for every label $\ell'$, 
$\weak{\lsupexp{\ell}{P},\ell'}$ is an upper $\omega$-continuous parametric prevision.
\end{proposition}

\begin{proof}
The fact that for every instruction $\lsupexp{\ell}{P}$, for every label $\ell'$, 
$\weak{\lsupexp{\ell}{P},\ell'}$ is $\omega$-continuous and monotonic follows directly from Lemma~\ref{weakScott}.
The measurability has just been proved in Lemma~\ref{weakmeasurable}.
It suffices to show the positive homogeneity and the "upper condition". We prove it by induction on instructions. 

\begin{itemize}
\item[$\bullet$] The identity is clearly a linear $\omega$-continuous prevision thus $\weak{\lsupexp{\ell}{\skipp},\ell'}$ 
is an upper parametric prevision.

\item[$\bullet$] Suppose, we have a map $g:\esigma\to \esigma$ and consider a map $F$ from $\mesp{\esigma}$ to 
itself defined by $F(f)=f\circ g$ for all $f\in\mesp{\esigma}$. The map $F$ is clearly a linear 
$\omega$-continuous prevision, this implies that $\weak{\lsupexp{\ell_1}{x:=e},\ell_2}$ is an upper parametric prevision.
 
\item[$\bullet$] By induction hypothesis, the maps $\weak{\lsupexp{\ell_{2}}{P},\ell_{4}}$ and $\weak{\lsupexp{\ell_{3}}{Q},\ell_{4}}$
are upper parametric previsions by Proposition~\ref{compo} (the third point) $\weak{\lsupexp{\ell_{1}}{P};\lsupexp{\ell_{2}}{Q},\ell_{3}}$
is an upper parametric prevision.

\item[$\bullet$] We use Proposition~\ref{equiv}. For all $\rho$ such that $\seme{t}(\rho)=\{0\}$, 
\[
\kappa\mapsto\weak{\lsupexp{\ell_1}{\ifp}\ t\ \thenp \lsupexp{\ell_2}{P}\ \elsep
\lsupexp{\ell_3}{Q},\ell_4}(\kappa)(\rho)=\left(\weak{\lsupexp{\ell_0}{P_0},\ell_2}(\kappa)\right)(\rho)
\]    
     which is by induction hypothesis a classical upper prevision. 
    For all $\rho$ such that $\seme{t}(\rho)=\{1\}$, the same argument leads to the result. 
    Now suppose that $\seme{t}(\rho)=\{0,1\}$, by Proposition~\ref{compo} (the second point), we conclude that 
    $\kappa\mapsto\weak{\lsupexp{\ell_1}{\ifp}\ t\ \thenp
    \lsupexp{\ell_2}{P}\ \elsep
    \lsupexp{\ell_3}{Q},\ell_4}(\kappa)(\rho)$ is, by induction hypothesis, a classical upper prevision. 
    The map $\kappa\mapsto\weak{\lsupexp{\ell_1}{\ifp}\ t\ \thenp
    \lsupexp{\ell_2}{P}\ \elsep
    \lsupexp{\ell_3}{Q},\ell_4}(\kappa)(\rho)$ is a classical upper prevision for all $\rho\in\esigma$ then
    $\weak{\lsupexp{\ell_1}{\ifp}\ t\ \thenp
    \lsupexp{\ell_2}{P}\ \elsep
    \lsupexp{\ell_3}{Q},\ell_4}$ is an upper parametric prevision.
\item[$\bullet$] By Proposition~\ref{compo} (the first point), it suffices to prove that the auxilary map 
$H_{t, \lsupexp{\ell_2}{P}, \ell_3}(\bot_{\mathbf{F}(\esigma,\ANS)})$
 is an upper $\omega$-continuous parametric prevision. From Lemma~\ref{wpprop}, 
 $H_{t, \lsupexp{\ell_2}{P}, \ell_3}(\bot_{\mathbf{F}(\esigma,\ANS)})$
 is $\omega$-continuous and monotonic. It suffices to show that 
 $H_{t, \lsupexp{\ell_2}{P}, \ell_3}(\bot_{\mathbf{F}(\esigma,\ANS)})$ 
 is positively homogeneous and upper. We prove the result by using Proposition~\ref{equiv}.
 Let $\rho\in\esigma$. Suppose that $\seme{t}(\rho)=\{1\}$.
The result follows from the induction hypothesis. Now suppose that $\seme{t}(\rho)=\{0\}$, 
 $H_{t, \lsupexp{\ell_2}{P}, \ell_3}(\bot_{\mathbf{F}(\esigma,\ANS)})$ is the identity and the result 
 follows from the linearity of the identity. Finally suppose that $\seme{t}(\rho)=\{0,1\}$, 
 the result follows from the stability of upper parametric prevision by binary suprema.
\end{itemize}
\end{proof}

\section{Special case of inputs}
\label{sec:input}

In this subsection, we are interested in interaction between the program and an external environment.
This interaction can be viewed as a sensor which saves data from the external environment thanks to a command \emph{input}.
We suppose that these data are at the same time noisy and imprecise. Mathematically, it can be modelled by 
$\omega$-capacities. It means that we want to represent for a fixed environment $\rho$ the \emph{input} as a $\omega$-capacity. 
We assume that only $k$ variables $x_{i_1},x_{i_2},\ldots,x_{i_k}$ are affected by the input.

A $\omega$-capacity on a topological space $X$ is a map $\nu:\bor{X}\mapsto \rr_+$ such that:
\[
\nu(\emptyset)=0,\ \nu(U)\geq 0\text{ and }\forall\, U\in\bor{X}\enspace .
\]

The $\omega$-capacity is said to be:
\begin{itemize}
\item monotonic iff $\forall\, U,V\in\bor{X}$:

\[
U\subseteq V \implies \nu(U)\leq \nu(V)\ ;
\]

\item continuous iff for all nondecreasing sequences $(U_n)_{n\in\nn}\subseteq \bor{X}$:

\[
\nu\left(\bigcupi_{n\in\nn} U_n\right)=\supi_{n\in\nn}\nu(U_n)\ ;
\]
{}
\item convex iff for all $U,V\in \bor{X}$:

\[
\nu\left(U\cup V\right)+\nu\left( U\cap V\right)\geq \nu(U)+\nu(V)
\ ;
\]

\item concave iff for all $U,V\in \bor{X}$:

\[
\nu\left(U\cup V\right)+\nu\left( U\cap V\right)\leq \nu(U)+\nu(V)
\ ;
\]

\end{itemize}
We will use the following result relying convexity and sub(super)linearity of the Choquet 
integrals.

\begin{proposition}
\label{additive}
Let $X$ be a topological space. Let $f$ and $g$ be in $\mesp{X}$. Let $\alpha$, $\beta$ two positive reals.

Let $\nu$ be a convex $\omega$-capacity, then the Choquet integral is superlinear:
\[
\choqint{x\in X} \alpha f(x)+\beta g(x)d\nu\geq \alpha\choqint{x\in X} f(x)d\nu+\beta\choqint{x\in X} g(x)d\nu
\]

Let $\mu$ be a concave $\omega$-capacity, then the Choquet integral is sublinear:
\[
\choqint{x\in X} \alpha f(x)+\beta g(x)d\mu\leq \alpha\choqint{x\in X} f(x)d\mu+\beta\choqint{x\in X} g(x)d\mu
\]
\end{proposition}

For $\rho\in\esigma$, we suppose that $\spesem{(x_{i_1},x_{i_2},\ldots,x_{i_k}=input( )}(\rho)$ is a monotonic continuous 
$\omega$-capacity $\nu$ on $\Va_{I}=\{x_{i_1},x_{i_2},\ldots,x_{i_k}\}$. We denote $\Va_{-I}=\{x\in\Va, x\notin \Va_{I}\}$ and 
we suppose that a certain $\rho_0:\Va_{-I}\mapsto \re$ (or with value in $\fe$) is given. We want to extend the $\omega$-capacity 
$\spesem{input}(\rho)$ to $\re$ (or $\fe$) with respect to the fact that the unaffected variables are represented by a fixed environment $\rho_0$.
We extend $\spesem{input}(\rho)$ to a $\omega$-capacity $\overline{\spesem{input}}(\rho)$ over $\esigma$ ($\simeq \re^{\Va}$ or $\simeq \fe^{\Va}$ ) as follows:

\[
\overline{\spesem{input}}(\rho)(C)=\spesem{input}(\rho)\left(\{x\in\re^{\Va_{I}} \mid (x,\rho_0)\in C\}\right)
\]
for all Borel sets $C$ of $\esigma$. This latter definition means that the measure of a Borel set
is completely determined by its affected part (by the instruction \emph{input}).

\begin{assum}
\label{assum:mesaure}
We assume that $\rho\mapsto \overline{\spesem{input}}(\rho)(U)$ is measurable for all $U\in\bor{\esigma}$.
\end{assum}

We define a last semantics which is the integration of "continuation" by a $\omega$-capacity.
Let $\kappa:\esigma\mapsto \brp$ be a positive measurable function. We define the semantics 
of the instruction \emph{input} as:

\[
\weak{input}(\kappa)(\rho):=\choqint{\rho'} \kappa(\rho')d\overline{\spesem{input}}(\rho)
\]

\begin{proposition}
Under the Assumption~\ref{assum:mesaure}, for all $\kappa\in\mesp{\esigma}$, 
the function $\rho\mapsto\weak{input}(\kappa)(\rho)$ belongs to $\mesp{\esigma}$.
\end{proposition}

\begin{proof}
The positivity is clear from the definition of the Choquet integral. We only give a proof 
for the measurability. Let $\kappa$ be a positive measurable function.  Then $\kappa$ is 
the nondecreasing supremum of a sequence of positive step functions $\left(\varphi_{n}\right)_{n\in\nn}$ and:
\[
\choqint{x\in X} \kappa(x)d\overline{ \spesem{input}}(\rho)=\choqint{x\in X}
\supi_{n\in\nn} \varphi_{n}(x)d \overline{\spesem{input}}(\rho)\enspace .
\]
From the $\omega$-Scott continuity of Choquet integrals, we get:
\[
\choqint{\rho'\in \esigma} \kappa(\rho')d \overline{\spesem{input}}(\rho)=\supi_{n\in\nn}\choqint{\rho'\in \esigma} \varphi_{n}(\rho')d \overline{\spesem{input}}(\rho)\enspace .
\]
The function $\varphi_{n}$ have the form $\alpha_{n}\sum_{i=0}^{K_{n}} \chi_{A_{i}^{n}}$ with 
$(A_{i}^{n})_{i}$ is a nonincreasing sequence of Borel sets for all $n\in\nn$. 
Thus, we have:
\[
\choqint{\rho'\in X} \varphi_{n}(\rho')d \overline{\spesem{input}}(\rho)=\alpha_{n}\sum_{i=0}^{K_{n}} \overline{\spesem{input}}(\rho)(A_{i}^{n})\enspace.
\]

Hence, from the Assumption~\ref{assum:mesaure}, the map $\rho\mapsto\overline{\spesem{input}}(\rho)(A_{i}^{n})$ is measurable for all $n\in\nn$, for all 
$i\in\{0,\ldots,K_{n}\}$ and then for all $n\in\nn$, $\rho\mapsto \alpha_{n}\sum_{i=0}^{K_{n}} \overline{\spesem{input}}(\rho)(A_{i}^{n})$ is measurable. 
We conclude that the map: $\displaystyle{\rho\mapsto \choqint{\rho'\in \esigma} \kappa(\rho')d\overline{\spesem{input}}(\rho)}$
is measurable since it is the pointwise supremum of measurable functions.
\end{proof}

\begin{proposition}
If the $\omega$-capacity $\spesem{input}(\rho)$ is convex (concave) and $\omega$-continuous 
then $\kappa\mapsto\weak{input}(\kappa)(\rho)$ defines a upper (lower) $\omega$-continuous 
prevision.  
\end{proposition}

The proof of this latter proposition is left to the reader. Indeed, from Proposition~\ref{additive},
if the capacity is convex (concave) then the Choquet integral is superlinear (sublinear). The proof 
is thus reduced to show that if the capacity $\spesem{input}(\rho)$ is convex or concave then the 
extended capacity $\overline{\spesem{input}}(\rho)$ fulfills the same property.

\bibliography{cpp}
\end{document}